\documentclass{endm}

\usepackage{endmmacro}
\usepackage{graphicx}
\usepackage{placeins}

\begin{document}

\begin{frontmatter}

\title{Tight lower bounds on the number of bicliques in false-twin-free graphs}

\vspace*{-10mm}

\author{Marina Groshaus\thanksref{myemail}}
\vspace*{-3.5mm}
\address{Departamento de Computaci\'on\\ Universidad de Buenos Aires\\ Buenos Aires, Argentina}

\vspace*{-2mm}

\author{Leandro Montero\thanksref{coemail}}
\vspace*{-3.5mm}
\address{Laboratoire de Recherche en Informatique\\ Universit\'e Paris-Sud\\ Orsay CEDEX, France}

\thanks[myemail]{Partially supported by UBACyT grant 20020100100754, PICT ANPCyT grant 2010-1970,
  CONICET PIP grant 11220100100310. Partially supported by Math-Amsud project 14 Math 06. Email: \href{mailto:groshaus@dc.uba.ar}{\texttt{\normalshape groshaus@dc.uba.ar}}} 
\thanks[coemail]{Email: \href{mailto:lmontero@lri.fr}{\texttt{\normalshape lmontero@lri.fr}}}

\vspace*{-5mm}

\begin{abstract}
A \emph{biclique} is a maximal bipartite complete induced subgraph of $G$. Bicliques have been studied in the last years motivated by the large 
number of applications. In particular, enumeration of the maximal bicliques has been of interest in data analysis.
Associated with this issue,  bounds on the maximum number of bicliques were given.
In this paper we study bounds on the minimun number of bicliques of a graph. Since adding false-twin vertices to $G$ does not change the number 
of bicliques, we restrict to false-twin-free graphs. We give a tight lower bound on the minimum number bicliques for a subclass of 
$\{C_4$,false-twin$\}$-free graphs and for the class of $\{K_3$,false-twin$\}$-free graphs. Finally we discuss the problem for general graphs. 
\end{abstract}

\begin{keyword}
Bicliques, False-twin-free graphs, Lower bounds
\end{keyword}

\end{frontmatter}

\vspace*{-2mm}

\section{Introduction}
Interest in the study of bicliques has increased recently motivated by the wide scope of 
applications~\cite{Amilhastre,Ganter,Kumar,Lehmann,Schweiger01072011,Xiang}. For example, bicliques appear in automata and language 
theory, graph compression, artificial intelligence, biology and data mining.  In particular, bicliques are studied in the contexts of: biclustering 
microarray data~\cite{cheng2000biclustering,tanay2002discovering,wang2002clustering}, optimizing phylogenetic tree 
reconstruction~\cite{Sanderson2003}, identifying common gene-set associations~\cite{Chesler2005}, integrating diverse 
functional genomics data~\cite{Baker2009}, analyzing proteome-transcriptome relationships~\cite{kirova2006systems} and discovering 
patterns in epidemiological research~\cite{mushlin2007graph}. In genetics, bicliques represent subsets of genes and subsets of 
properties~\cite{cheng2000biclustering,liu2003op,tanay2002discovering,wang2002clustering}. Also, bicliques appear in the reconstruction of 
phylogenetic trees~\cite{Sanderson2003}.

In all these applications, bicliques represent the relation between different data types. Due to the big size of data used, 
one of the main problems consists of enumerating all the bicliques. Recall that the number of bicliques in a graph can be 
exponential~\cite{PrisnerC2000}.
Therefore, algorithms to list all bicliques of a graph have been of particular attention. Recall that the definition of bicliques can 
be different for each case. Some authors consider bicliques not induced, others, with bounded size of bipartition, etc.
See for example~\cite{Alexe,Damaschke2014317,Dias,Eppstein1994,Gely,Makino2004,Mukherjee,mushlin2007graph,Sanderson2003,Raj,zaki1998theoretical,Zhang2014BIO}.

In the context of induced bicliques, some bounds were given. Prisner studied various aspects~\cite{PrisnerC2000}. He showed that the 
maximum number of bicliques in a bipartite graph on $n$ vertices is $2^{\frac{n}{2}}$, that is attained by the $Cocktail-party$ graphs.
Also, he gave a lower bound of $3^{\frac {n}{3}}$ and an upper bound of $1.6181^n$ on the maximum number of 
bicliques for general graphs. In~\cite{Kratsch}, they improved this result giving a better upper bound of 
$\frac{1}{3^{1/3}-1}3^{\frac {n}{3}}$ on the number of bicliques of a graph.

In this work we focus on lower bounds on the minimum number of bicliques of a graph where by definition, bicliques are induced.  Despite in general it seems  a ``trivial'' question, 
since any bipartite complete graph $K_{n,m}$ has just one biclique, it is not considered the fact that all vertices in each 
partition are false-twins. Since adding false-twin vertices to $G$ does not change the number of bicliques~\cite{marinayo}, the 
bipartite complete graph $K_{n,m}$ can be thought as a biclique ``equivalent'' to $K_{1,1}$. Recall that many applications consider 
the bicliques as a whole group of items (represented in one side of the bipartition) having a common non-empty subset of characteristics 
(represented in the other side). So, it is natural to think that if there are groups of ``twin objects'' with the same characteristics, 
we could only maintain one representative object for each. Following this idea we introduce the problem of finding bounds for the minimum 
number of bicliques in a graph with no false-twin vertices.  We recall that deleting false-twin vertices can be done in linear 
time using the modular decomposition~\cite{Habib}. 

Even though the bounds are polynomial, this approach can help to develop algorithms for listing the bicliques, not only because of 
the use itself of the bounds but also because of the ideas behind the proofs. Also these bounds can be used for heuristical algorithms. 

Throughout this paper we discuss the general case and give bounds for the minimum number of bicliques of graphs in a subclass of 
$\{C_4$, false-twin$\}$-free graph and the class of $\{K_3$,false-twin$\}$-free graph. We prove that any graph in these classes on 
$n\geq 3$ and $n\geq 4$ vertices respectively, has at least $\lceil \frac{n}{2} \rceil$ bicliques.
We show that these bounds are tight and finally we prove that the bound $\lceil \frac{n}{2} \rceil$ does not hold for general 
false-twin-free graphs. This work is the full improved version of a previous extended abstract~\cite{LAGOSENDM}.

This work is organized as follows. 
In Section $2$ the notation is given. In Section $3$ and Section $4$ we study bicliques in false-twin-free graphs.
We prove the lower bound for the minimum number of bicliques for a subclass of $\{C_4$,false-twin$\}$-free graphs and for the class of 
$\{K_3$,false-twin$\}$-free graphs. Finally, in Section $5$ we discuss the problem for general graphs. 

\section{Preliminaries}

Along the paper we restrict to undirected simple graphs. Let $G=(V,E)$ be a graph with vertex set $V$ and edge set $E$, and 
let $n=|V|$ and $m=|E|$.  A \emph{subgraph} $G'$ of $G$ is a graph $G'=(V',E')$ where $V'\subseteq V$ and 
$E'\subseteq E$. A subgraph $G'=(V',E')$ of $G$ is \emph{induced} when for every pair of vertices $v,w \in G'$, $vw \in E'$ 
if and only if $vw \in E$. A graph $G=(V,E)$ is \emph{bipartite} when $V= U \cup W$, $U \cap W = \emptyset$, and $E \subseteq U \times W$. Say that $G$ is a 
\emph{complete graph} when every possible edge belongs to $E$ and say that $G$ is \emph{bipartite complete} when 
$E=U\times W$. A complete graph of $n$ vertices is denoted $K_{n}$ and
a bipartite complete graph on $n$ and $n'$ vertices in each partition respectively, is denoted $K_{n,n'}$.
A graph $G$ is \emph{$H$-free} if it does not contain $H$ as an induced subgraph.
An \emph{independent set} of a graph is a set of vertices, no two of which are adjacent.
A \emph{clique} of $G$ is a maximal complete induced subgraph, while a \emph{biclique} is a maximal bipartite complete induced subgraph of $G$. 
The \emph{open neighborhood} of a vertex $v \in V(G)$, denoted by $N(v)$, is the set of vertices adjacent to $v$ while the \emph{closed neighborhood} of 
$v$, denoted by $N[v]$, is $N(v) \cup \{v\}$. The \emph{degree} of a vertex $v$, denoted by $d(v)$, is defined as $d(v) = |N(v)|$.
The maximum degree among all vertices of $G$ is denoted by $\Delta(G)$.
Two vertices $u$, $v$ are \emph{false-twins} if $N(u)=N(v)$ and \emph{true-twins} if $N[u]=N[v]$.  A vertex $v$ is \emph{simplicial} if 
$\{v\} \cup N(v)$ is a clique. A \emph{diamond} is the graph $K_4$ minus an edge. A \emph{cycle} on $k$ vertices, denoted by $C_k$, is a set of 
vertices $v_{1}v_{2}\ldots v_{k} \in V(G)$ such that $v_{i} \neq v_{j}$ for all $1 \leq i \neq j \leq k$, $v_{i}$ is adjacent to $v_{i+1}$ 
and $v_k$ is adjacent to $v_1$ for all $1 \leq i \leq k-1$.  We assume that all the graphs of this paper are connected.

As mentioned before, we want to keep only the representative vertices for each group of false-twins, in order to, for example reduce the size of the 
graph and bicliques considered. Formalizing this idea, we present the following definition as it is in~\cite{marinayo}. Consider all maximal sets of 
false-twin vertices $Z_1,\ldots,Z_k$ and let $\{z_{1},z_{2},\ldots,z_{k}\}$ be the set of \textit{representative vertices} such that $z_i\in Z_i$.  
The graph obtained by the deletion of all vertices of $Z_i\setminus \{z_i\}$, for $i=1,\ldots,k$, is denoted \textit{Tw(G)}.  
We mention that being false-twin can be thought as an equivalence relation that separates the graphs in equivalent classes.

\section{Lower bounds in a subclass of $\{C_4$,false-twin$\}$-free graphs}

We start defining a subclass of $\{C_4$,false-twin$\}$-free graphs. Observe that if $G$ has no induced $C_4$, then all its bicliques 
are isomorphic to $K_{1,r}$, $r\geq 1$. We call these kinds of bicliques \emph{$v$-star}, where $v$ is the vertex in the partition of size one 
and is called the \emph{center} of the star.

A vertex $x$ is \emph{alone} if $x$ is a simplicial vertex and $N(x)$ has no simplicial vertices.
Let $G$ be a graph and let $\mathcal{A}$ be its set alone vertices. An \emph{assignment} for $\mathcal{A}$ is an association for every 
 alone vertex: each alone vertex $x$ is associated to a vertex $v$ and an edge $vv'$, where $v$ and $v'$ belong to $N(x)$, such that either 
 $v$ is not dominated by $v'$ or $v$ and $v'$ are true-twins.  
We say that $\mathcal{A}$ has a \emph{good assignment} if the assignment verifies the following:
If two different alone vertices in $\mathcal{A}$ have the same associated edge, then its endpoints are not true-twins and the 
associated vertices of the alone vertices are different. 
This implies that the edges incident to true-twins or to two vertices such that one is dominated by the other, can only be assigned to at 
most one alone vertex. See Figures~\ref{claselocaIN} and~\ref{claselocaOUT} for examples.

\FloatBarrier
\begin{figure}[ht!]
	\centering
	\includegraphics[scale=.5]{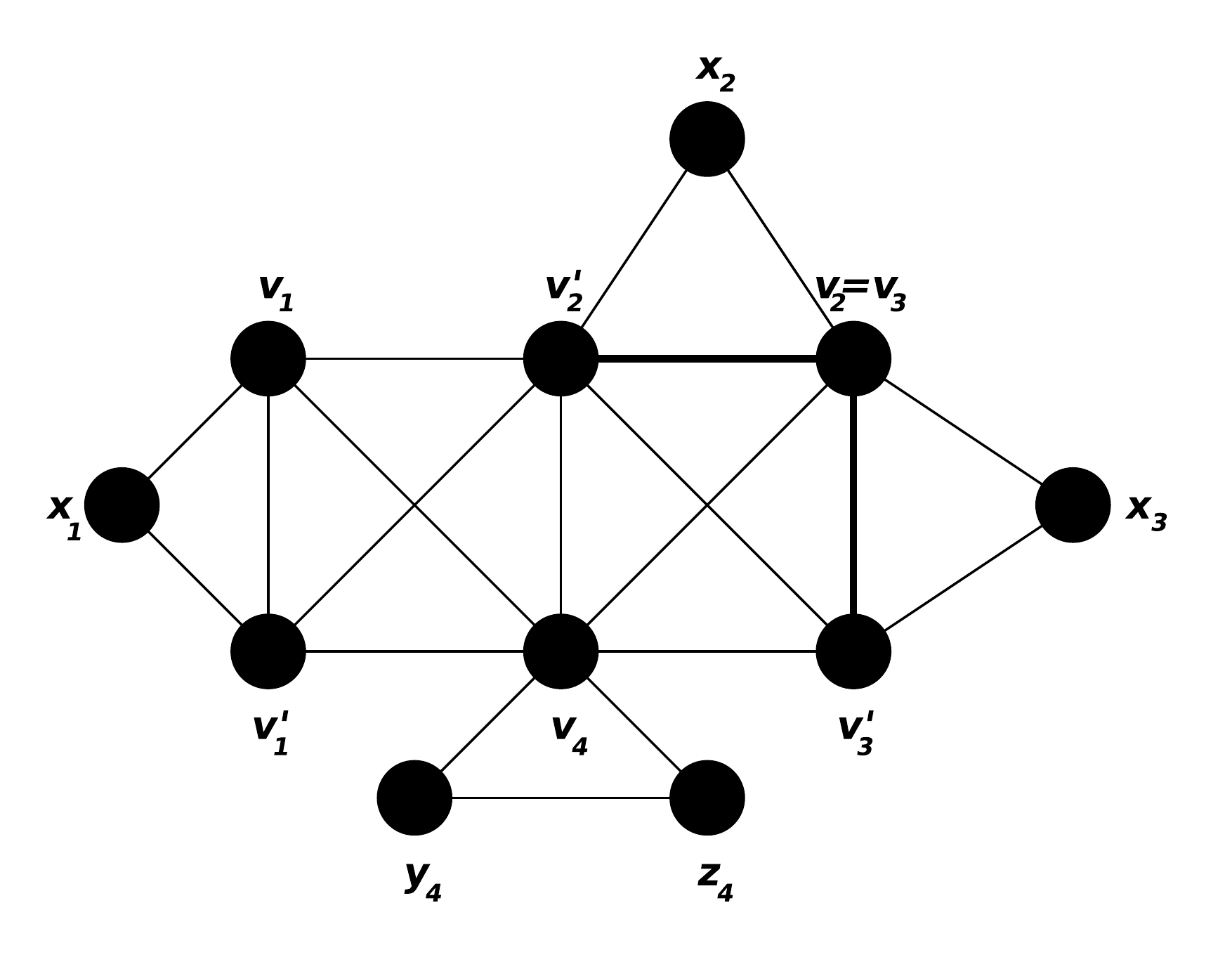}
	\caption{$\{C_4$,false-twin$\}$-free graph that has a good assigment. Vertex $x_1$ has the vertex $v_1$ and an edge with true-twins as 
	endpoints associated to it ($v_1v_1'$). Vertices $x_2$ and $x_3$ have the same vertex ($v_2=v_3$) associated to them but the 
	associated edges are different,	$v_2v_2'$, $v_3v_3'$ respectively.}
	\label{claselocaIN}
	\includegraphics[scale=.5]{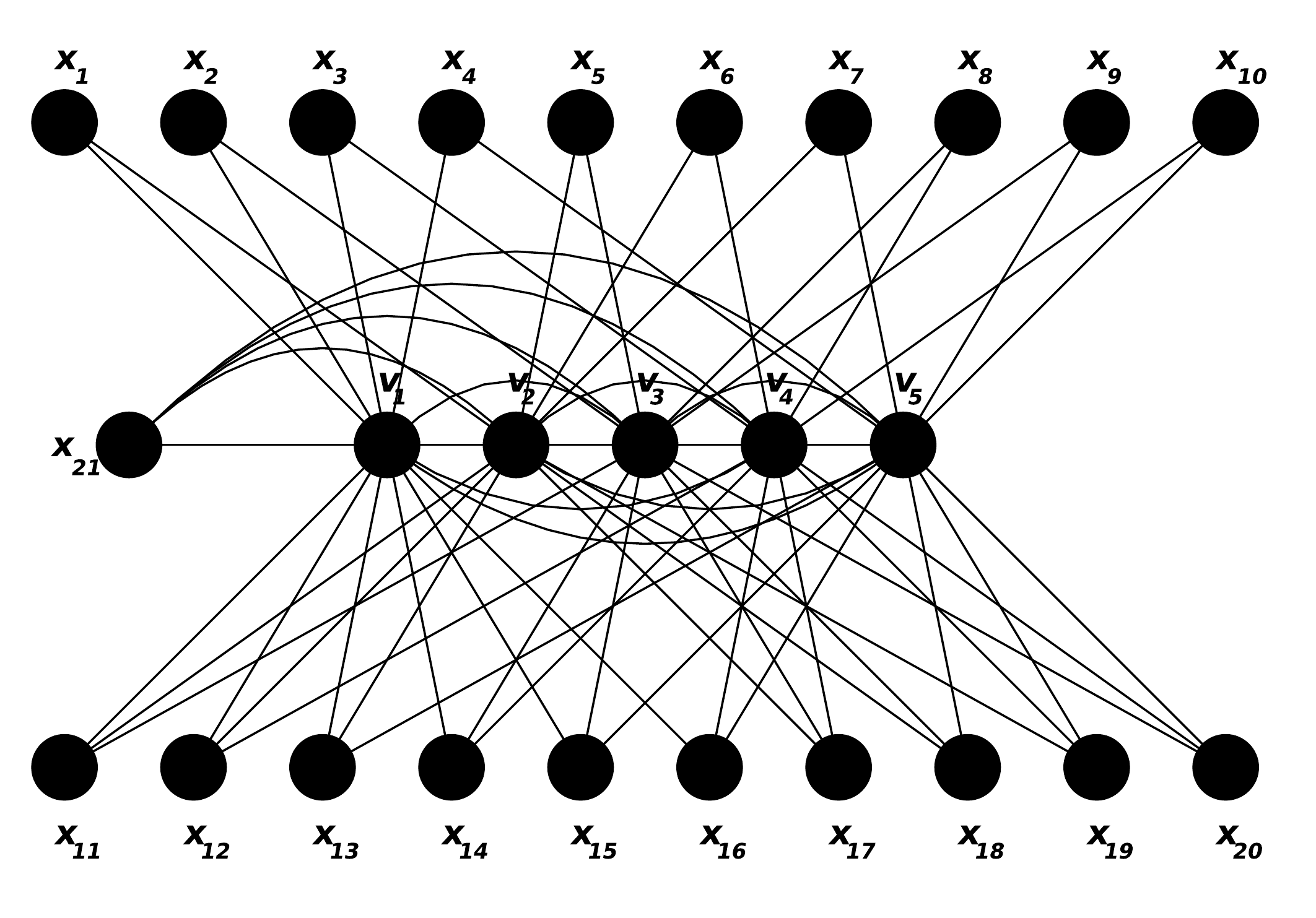}
	\caption{$\{C_4$,false-twin$\}$-free graph that has not a good assigment. Vertices $x_1,\ldots,x_{21}$  are alone vertices 
	and $v_1,\ldots,v_5$ are non-simplicial vertices. Therefore in any assigment at least two alone vertices will be associated
	to the same vertex and edge as each of the $10$ edges in the clique $v_1,\ldots,v_5$ could be used twice, obtaining at most $20$ 
	possible assigments to associate to $21$ vertices.}
	\label{claselocaOUT}
\end{figure}
\FloatBarrier

We study the subclass of $\{C_4$,false-twin$\}$-free graphs such that the set of alone vertices $\mathcal{A}$ has a good assignment. 
Observe that $\{C_4$,diamond,false-twin$\}$-free graphs are included in this class.   

We need the following lemma.

\begin{lemma}\label{singrado1}

Let $G$ be a $\{C_4$,false-twin$\}$-free graph on $n\geq 3$ vertices, without vertices of degree one, 
such that the set $\mathcal{A}$ has a good assignment. Then $G$ has at least $n$ bicliques.
\end{lemma}
\begin{proof}
Observe that if $G=K_n$, then the result holds as each edge is a biclique. Otherwise, we will assign each vertex to a different 
biclique. We will first give labels to some specific edges and associate them to some vertices. Also we will assign some vertices to some 
bicliques.

Let $C_1,C_2,\ldots,C_{\ell}$, $\ell \geq 1$, be the maximal sets of pairwise adjacent simplicial vertices in $G$. 
Observe that each $C_i$ induces a complete subgraph. 

Consider first the sets $C_i$ such that $|C_i|\geq 3$. For each of these sets choose any non-simplicial vertex $v$ adjacent to two 
vertices in $C_i$, say $w_1,w_2$. Clearly $v$ exists as $G \neq K_n$. Label the edges $vw_1, vw_2$ with labels $1$, $2$ respectively. 
This pair of labelled edges are associated to $v$. 
Now as each of the $\frac{|C_i|(|C_i|-1)}{2}$ edges of $C_i$ is a biclique in $G$ and $\frac{|C_i|(|C_i|-1)}{2} \geq |C_i|$, 
for $|C_i|\geq 3$, we assign each of these vertices to a different biclique (edge) of the complete graph $C_i$. 

Next consider the sets $C_i$ such that $|C_i|=2$. Assume $C_i=\{y_i, z_i\}$. For each set, choose any non-simplicial 
vertex $v$ adjacent to $y_i$ and $z_i$. Consider the pair of edges $vy_i, vz_i$ and assign labels $1$ and $2$ respectively.  
This pair of labelled edges are associated to $v$. Assign vertex $y_i$ to the biclique $y_iz_i$.

Finally, consider the sets $C_i$ such that $|C_i|=1$. Assume $C_i=\{x_i\}$. Recall that $x_i$ is an alone 
vertex. By hypothesis, $\mathcal{A}$ has a good assignment. For each $i$, consider the vertex $v$ and the edge 
$vv'$ associated to $x_i$ in the good assignment. 
If $v$ and $v'$ are true-twins, $vv'$ is a biclique, and we assign vertex $x_i$ to this biclique. 
Observe that since $\mathcal{A}$ has a good assigment, no two vertices are associated to the same edge (biclique). 
Otherwise, label the edge $vx_i$ with label $1$, and the edge 
$vv'$ with label $2$. This pair of labelled edges are associated to $v$. 

Observe that each pair of edges labelled with $1$ and $2$ contain a common vertex and are also incident to two adjacent vertices. 
See Figure~\ref{claselocalabel}.

Now, we assign the remaining vertices to different bicliques. 

\begin{figure}[ht!]
	\centering
	\includegraphics[scale=.5]{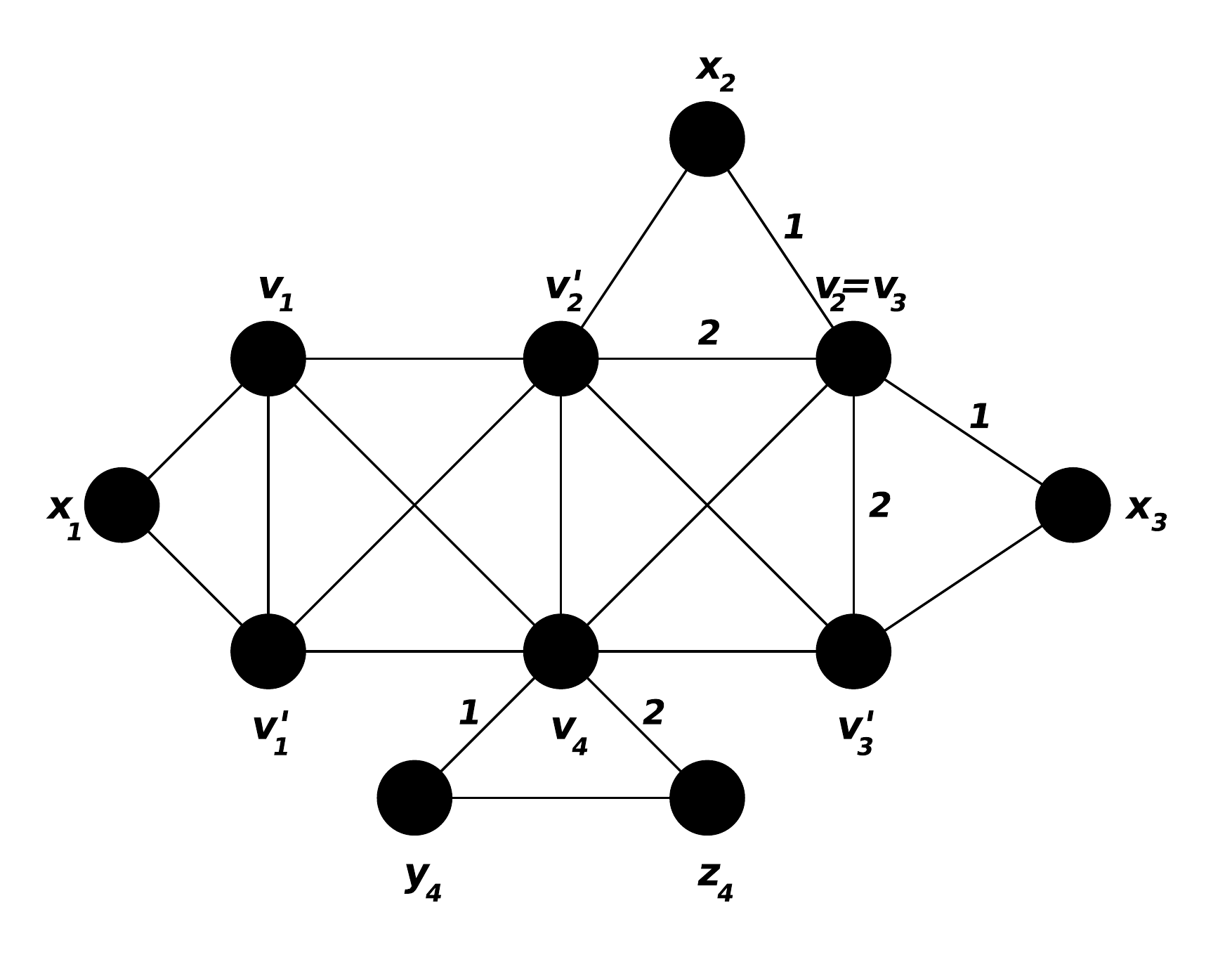}
	\caption{Same graph as Figure~\ref{claselocaIN} with labelled edges.}
	\label{claselocalabel}
\end{figure}

First consider non-simplicial vertices. 
Let $v$ be a non-simplicial vertex. Observe that each $v$ is the center of star biclique. 
Therefore if $v$ has no labelled edges associated to it, then we assign $v$ to any $v$-star biclique.
Otherwise, consider any star biclique centered in $v$ that contains all the associated edges with label $1$. 
Recall that all these edges labelled with $1$ are incident to $v$ and to simplicial vertices which are all independent. 
Also, observe that any star containing all those edges labelled with $1$ incident to $v$ does not contain a labelled edge associated to 
$v$ with label $2$. So such a vertex $v$ is associated with a star biclique centered at $v$, with all associated edges labelled $1$.

Next, for each set of simplicial vertices $C_i=\{y_i, z_i\}$ with edges $vy_i, vz_i$ labelled with $1$, $2$ respectively, we will assign $z_i$ to a biclique. Consider any  
$v$-star biclique that contains the edge $vz_i$ with label $2$ and all the edges associated to $v$ labelled with $1$ 
(different from $vy_i$), if any. Observe that such a $v$-star biclique containing the edge $vz_i$ always exists since, if $v$ has no other 
labelled edges associated to it, as it is non-simplicial, then it must exist a vertex adjacent to $v$ not adjacent to $z_i$.
Assign vertex $z_i$ to this biclique.

Finally, consider each $C_i=\{x_i\}$. Recall that $x_i$ is an alone vertex. Let $v$ and $vv'$ be its 
corresponding vertex and edge in the assignment. Vertex $v$ is not dominated by $v'$ since the case when they are true-twins was solved before. 
Consider any $v$-star biclique that contains the 
labelled edge $vv'$ with label $2$ associated to $v$ and no other labelled edge associated to $v$ with label $2$. 
Observe that such a star always exists since $v$ is not dominated by $v'$ and if a labelled edge with label $1$ associated to $v$ cannot be in the 
star, neither can be the corresponding labelled edge with label $2$ of the same pair.  
We assign vertex $x_i$ to this biclique.

We have assigned every vertex of $G$ to a biclique of $G$. Observe that all these bicliques are different:
each non-simplicial vertex $v$ is assigned to a $v$-star biclique that contains no labelled 
edge associated to $v$ with label $2$. Each simplicial vertex is assigned to a biclique that is either an edge incident to two simplicial 
vertices, an edge incident to two non-simplicial true-twins vertices, or a star centered in a 
non-simplicial vertex that contains exactly one associated labelled edge with 
label $2$.  We conclude that $G$ contains at least $n$ bicliques. 
\end{proof}

\begin{theorem}\label{bicsNoC4general}
Let $G$ be a $\{C_4$,false-twin$\}$-free graph on $n\geq 3$ vertices, such that $G$ has at most $k$ vertices of degree one and the set 
$\mathcal{A}$ of alone vertices has a good assignment. 
Then $G$ has at least $n-k$ bicliques.
\end{theorem}
\begin{proof}
Clearly, if $G$ does not contain vertices of degree one, the result holds by Lemma~\ref{singrado1}.
Let $v_1,v_2,\ldots,v_k$ be the vertices of degree one in $G$, $k\geq 1$, and let $x_1,x_2,\ldots,x_k$ be their unique neighbors respectively. 
We obtain from $G$ a new graph $G'$ as follows. For each vertex $v_i$, we add two new vertices $u_i,w_i$ such that $\{v_i,u_i,w_i\}$ induces 
a $K_3$. Clearly $G'$ is $\{C_4$,false-twin$\}$-free and it has $n+2k$ vertices where none of them has degree one and the set 
$\mathcal{A}$ of alone vertices has a good assignment. Then, by Lemma~\ref{singrado1}, $G'$ has at least $n+2k$ bicliques. Now for each 
$i$, $\{u_i,w_i\}$, $\{v_i,u_i,x_i\}$ and $\{v_i,w_i,x_i\}$ are bicliques in $G'$ that clearly, are 
not in $G$. Thus $G'$ has exactly $3k$ bicliques more than $G$. Finally we obtain that $G$ has at least $n+2k-3k = n-k$ bicliques as desired.
\end{proof}

Using the following Lemma, we can give a bound on the number $n-k$ of bicliques given by Theorem~\ref{bicsNoC4general}.

\begin{lemma}\label{grado1}
Let $G$ be a false-twin-free graph on $n\geq 3$ vertices. Then $G$ has at most $\lfloor \frac{n}{2} \rfloor$ vertices of degree one.
\end{lemma}
\begin{proof}
By contrary, if $G$ has more than $\lfloor \frac{n}{2} \rfloor$ vertices of degree one, then there must exist two of them having the same 
unique neighbor, that is, they are false-twin vertices. A contradiction.
\end{proof}

Combining last results, we obtain the main theorem of the section. It gives a tight lower bound on the number of bicliques for this subclass of  
$\{C_4$,false-twin$\}$-free graphs.

\begin{theorem}\label{bicsNoC4}
Let $G$ be a $\{C_4$,false-twin$\}$-free graph on $n\geq 3$ vertices such that the set $\mathcal{A}$ of  alone vertices has a 
good assignment. 
Then $G$ has at least $\lceil \frac{n}{2} \rceil$ bicliques.
\end{theorem}
\begin{proof}
The results follows from Lemma~\ref{grado1} and Theorem~\ref{bicsNoC4general} since 
$n-k \geq n - \lfloor \frac{n}{2} \rfloor = \lceil \frac{n}{2} \rceil$ as desired.
\end{proof}

This bound is tight. For this, consider any cycle $C_k=v_1v_2\ldots v_k$, $k \geq 5$, and join each $v_i$ with a new vertex $x_i$.
Clearly this graph is $\{C_4$,diamond,false-twin$\}$-free, it has $n=2k$ vertices and $\lceil \frac{n}{2} \rceil = k$ bicliques.

As a direct consequence of Theorem~\ref{bicsNoC4}, we obtain the following corollaries.

\begin{corollary}\label{bicsarbol}
Let $T$ be a false-twin-free tree on $n\geq 4$ vertices. Then $T$ has at least $\lceil \frac{n}{2} \rceil$ bicliques.
\end{corollary}

Moreover,

\begin{corollary}
For $\lceil \frac{n}{2} \rceil \leq k \leq n-2$ and $n \geq 4$ there exists a false-twin-free tree $T$ on $n$ vertices and $k$ bicliques.
\end{corollary}

\section{Lower bounds in $\{K_3$,false-twin$\}$-free graphs}

In this section we study bounds on the minimum number of bicliques in $\{K_3$,false-twin$\}$-free graphs. 
We show first some useful lemmas.

\begin{lemma}\label{k1rsiempre}
Let $G$ be a $\{K_3$,false-twin$\}$-free graph. Then every vertex is contained in a $v$-star biclique.
\end{lemma}
\begin{proof}
Let $v$ be a vertex. If $|N(v)|=1$ then the result clearly follows. Suppose now that $|N(v)| > 1$. Now, since $N(v)$ is an independent set, 
$\{v\} \cup N(v)$ is contained in one biclique. If there is no vertex $w$ such that $N(v) \subseteq N(w)$ then $\{v\} \cup N(v)$ is a $v$-star biclique. 
Otherwise, let $u$ be the vertex with maximum degree among all vertices in $N(v)$. Clearly, since $G$ is $\{K_3$,false-twin$\}$-free, 
if there are two vertices of same maximum degree, then they have some different neighbors. Hence, $\{u\} \cup N(u)$ is a $u$-star biclique 
that contains the vertex $v$ as desired.
\end{proof}

Based on the proof of last lemma, we obtain this immediate result.

\begin{corollary}
Let $G$ be a false-twin-free graph. Let $v$ be a vertex such that $d(v)=\Delta(G)$ and $v$ does not belong to a $K_3$. 
Then $\{v\} \cup N(v)$ is a biclique.
\end{corollary}

The next result will help us to prove the main theorem of the section.

\begin{lemma}\label{vkbics}
Let $G$ be a $\{K_3$,false-twin$\}$-free graph. Let $v$ be a vertex such that $d(v)=k$. Then $v$ belongs to at least $k$ different bicliques.
\end{lemma}
\begin{proof}
Let $v_1,v_2,\ldots,v_k$ be the neighbors of $v$. Clearly they are an independent set. Let $x_1,x_2,\ldots,x_\ell$ be the vertices adjacent to  
$v_1,v_2,\ldots,v_k$ (not including $v$). Let $G'$ be the subgraph induced by $\{v\} \cup \{v_1,v_2,\ldots,v_k\} \cup \{x_1,x_2,\ldots,x_\ell\}$. 
Clearly, $v_1,v_2,\ldots,v_k$ are not false-twins in $G'$ and since $G$ is $K_3$-free, $v$ is not adjacent to any $x_j$, $1 \leq j \leq \ell$. 
Now, for each $1 \leq i \leq k$, let $S_{v_i} = \{N_{G'}(x_j) : v_i \in N_{G'}(x_j), 1 \leq j \leq \ell\} \cup \{N(v)\}$.
Let $Cl(S_{v_i}) = \bigcap_{S \in S_{v_i}}S$. Observe first that $S_{v_i} \neq \emptyset$ for all $i$ since $N(v)$ belongs to all of them.
Observe then that $Cl(S_{v_i}) \cup (\{x_j : N_{G'}(x_j) \in S_{v_i}, 1 \leq j \leq \ell\ \} \cup \{v\})$ is a biclique in $G'$ and therefore a biclique in $G$. 
We show now that $Cl(S_{v_i}) \neq Cl(S_{v_j})$ for all $i\neq j$, i.e., $v$ belongs to $k$ different bicliques in $G$. 
Suppose by contrary, that $Cl(S_{v_i}) = Cl(S_{v_j})$. 
Now, since $v_i \in Cl(S_{v_i})$, we have $v_i \in Cl(S_{v_j})$. Similarly, $v_j \in Cl(S_{v_i})$. So, for all 
$S \in S_{v_i}$, we have $v_j \in S$. Also, for all $S \in S_{v_j}$, we have $v_i \in S$. That is, $N(v_i) = N(v_j)$, a contradiction since 
$G$ has no false-twin vertices.
\end{proof}

As a corollary, we obtain the following.

\begin{corollary}\label{kbicsperdidas}
Let $G$ be a $\{K_3$,false-twin$\}$-free graph. Suppose that there is a vertex $v$ such that $G-\{v\}$ has $k$ sets of false-twin vertices.
Then $G-\{v\}$ has at least $k$ bicliques less than $G$.
\end{corollary}
\begin{proof}
Observe first that, since $G$ has no false-twin vertices, every set of false-twin vertices in $G-\{v\}$ has size exactly two.
Let $\{v_i,w_i\}$ be the $k$ sets of false-twin vertices, such that $v$ is adjacent to $v_i$, $i=1,\ldots,k$. Observe now that, since $v_i$ and $w_i$ 
are false-twins in $G-\{v\}$, they belong to exactly the same bicliques but those bicliques containing the edge $vv_i$. Consider now the subgraph induced by the 
vertices $\{v\} \cup \{v_1,v_2,\ldots,v_k\} \cup N(v_1) \cup \cdots \cup N(v_k)$. Call this graph $G'$. Clearly, $v_1,v_2,\ldots,v_k$ are not 
false-twins in $G'$. Now, by Lemma~\ref{vkbics}, $v$ belongs to $k$ different bicliques in $G'$. These bicliques in $G'$ are either bicliques 
or are contained in bigger bicliques in $G$, but they do not contain any of the vertices $w_i$. Therefore, after removing $v$, these $k$ bicliques 
are lost in $G-\{v\}$ since any other biclique containing any $v_i$ contains also $w_i$. 
\end{proof}

Combining the last three results, it follows the main theorem of the section. It gives a tight lower bound on the number of bicliques for 
$\{K_3$,false-twin$\}$-free graphs.

\begin{theorem}\label{bicsNoK3}
Let $G$ be a $\{K_3$,false-twin$\}$-free graph on $n\geq 4$ vertices. Then $G$ has at least $\lceil \frac{n}{2} \rceil$ bicliques.
\end{theorem}
\begin{proof}
The proof is by induction on $n$. For $n=4$ the result trivially holds. Suppose $n\geq 5$. Now, by Lemma~\ref{k1rsiempre} there is a vertex $v$ contained 
in a star biclique. Without loss of generality, we can suppose that $v$ is the center. 
Consider the graph $G'=G-\{v\}$. We have the following two cases.
\begin{itemize}
 \item $G'$ is disconnected. Let $G_1,G_2,\ldots,G_s$ be the connected components of $G'$ on $n_1,n_2,\ldots,n_s$ vertices respectively. Since $G$ has no 
false-twin vertices, there can exist at most one $G_i$ such that $n_i=1$. Suppose that there are $\ell$ components, 
$\{G_{i_1},G_{i_2},\ldots,G_{i_\ell}\} \subseteq \{G_1,G_2,\ldots,G_s\}$ on $n_{i_1},n_{i_2},\ldots,n_{i_\ell}$ vertices respectively, such that each $G_{i_j}$ has 
$k_{i_j}$ sets of false-twin vertices, $j=1,\ldots, \ell$. It is easy to see that each of these sets has exactly two vertices, otherwise $G$
would have false-twin vertices. Now, by Corollary~\ref{kbicsperdidas}, $G'$ has at least $k_{i_1}+k_{i_2}+\cdots+k_{i_\ell}$ bicliques less than $G$.
Also, since $G'$ is disconnected, $v$ along with at least one vertex of each of the $s$ components is a biclique $B$ in $G$ 
isomorphic to $K_{1,r}$ that clearly, is not a biclique in $G'$. Now, consider for each $G_{i_j}$ the graph 
$Tw(G_{i_j})$. Each of these graphs have $n_{i_j}-k_{i_j}$ vertices and no false-twin vertices. If $n_{i_j}-k_{i_j}=2$ then $Tw(G_{i_j})=K_2$, and 
therefore it has one biclique, i.e., at least $\Big\lceil \frac{n_{i_j}-k_{i_j}}{2} \Big\rceil$ bicliques. 
Note that $n_{i_j}-k_{i_j} \neq 3$ as $Tw(G_{i_j})$ is also $K_3$-free. If $n_{i_j}-k_{i_j} \geq 4$, 
by the inductive hypothesis, $Tw(G_{i_j})$ has also at least $\Big\lceil \frac{n_{i_j}-k_{i_j}}{2} \Big\rceil$ bicliques. Now, for all other $G_i$ 
without false-twin vertices, if $n_i=2$, $G_i$ has, as before, $1=\lceil \frac{n_i}{2} \rceil$ biclique, $n_i=3$ is impossible as $G_i$ is
$\{K_3$,false-twin$\}$-free and for $n_i \geq 4$, by the inductive hypothesis,
$G_i$ has at least $\lceil \frac{n_i}{2} \rceil$ bicliques. If we sum up everything (and suppose the worst case, that is, there exists one $G_i$, 
say $G_s$, such that $n_i=1$), then the number of bicliques of $G$ is at least 
\begin{displaymath}
\bigg(\sum\limits_{j=1}^{\ell}\bigg\lceil \frac{n_{i_j}-k_{i_j}}{2} \bigg\rceil + k_{i_j}\bigg) + 
\bigg(\sum\limits_{i=1,i\neq i_j}^{s-1}\bigg\lceil \frac{n_i}{2} \bigg\rceil\bigg) + 1 \geq 
\end{displaymath}
\begin{displaymath}
\bigg(\sum\limits_{j=1}^{\ell}\bigg\lceil \frac{n_{i_j}}{2} \bigg\rceil \bigg) + 
\bigg(\sum\limits_{i=1,i\neq i_j}^{s-1}\bigg\lceil \frac{n_i}{2} \bigg\rceil\bigg) + 1 \geq
\bigg(\sum\limits_{i=1}^{s-1}\bigg\lceil \frac{n_i}{2} \bigg\rceil\bigg) + 1 \geq
\Big\lceil \frac{n}{2} \Big\rceil
\end{displaymath}
as desired.

\item $G'$ is connected. Suppose first that in $G'$ there are $k$ sets of false-twin vertices. As before, each of these sets has two
vertices. Then, by Corollary~\ref{kbicsperdidas}, $G'$ has $k$ bicliques less than $G$. Consider now the graph $Tw(G')$. This graph has $n-k-1 \geq 4$ 
vertices (or just two vertices, i.e., one biclique. Remark that three vertices is not possible.) and no false-twin vertices, therefore we can 
apply the inductive hypothesis. We conclude that 
$Tw(G')$ has at least $\Big\lceil \frac{n-k-1}{2} \Big\rceil$ bicliques. Then, $G$ has at least 
$\Big\lceil \frac{n-k-1}{2} \Big\rceil + k \geq \lceil \frac{n}{2}\rceil$ 
bicliques as desired. Suppose last that $G'$ has no false-twin vertices. By the inductive hypothesis, $G'$ has at least $\lceil \frac{n-1}{2}\rceil$ 
bicliques. Finally, since the $v$-star biclique is not in $G'$, we conclude that $G$ has at least 
$\lceil \frac{n-1}{2}\rceil + 1 \geq \lceil \frac{n}{2}\rceil$ bicliques.
\end{itemize}
Since we covered all cases the proof is now complete.
\end{proof}

Clearly this bound is tight as the same family of graphs presented in Section $3$ is $\{K_3$,false-twin$\}$-free.

As a consequence of Theorem~\ref{bicsNoK3}, we have the following.

\begin{corollary}\label{bicsbip}
Let $G$ be a false-twin-free bipartite graph on $n\geq 4$ vertices. Then $G$ has at least $\lceil \frac{n}{2} \rceil$ bicliques.
\end{corollary}

\section{Discussion for bounds in false-twin-free graphs}

In this section we address the following question:  Is it true
that every false-twin-free graph $G$ has at least $\lceil \frac{n}{2} \rceil$ bicliques ?

We answer this question showing a family of false-twin-free graphs with $k + 2^k -1$ vertices and $k^2$ bicliques.
Consider a graph $G$ constructed as follows. Take a clique $K=\{v_1,v_2,\ldots,v_k\}$ on $k$ vertices and an independent 
set $I=\{w_1,w_2,\ldots,w_{2^k-1}\}$. Consider the set of subsets $\mathcal{B}=\mathcal{P}(K)-\{\emptyset\}$, that is, the power set of $K$ 
minus the subset containing the empty set. Clearly $|\mathcal{B}|=2^k-1$ and all its subsets are different. Finally set $N(w_i) = B_i$, 
for $B_i \in \mathcal{B}$, $i=1,\ldots,2^k-1$.
It is easy to see that the graph $G$ constructed in this way has $k + 2^k -1$ vertices and no false-twins. Moreover, it has no induced $C_4$ 
therefore all its bicliques are stars. See Figure~\ref{contraejGraph}.

\begin{figure}[ht!]
	\centering
	\includegraphics[trim=0 0 0 0, scale=.5]{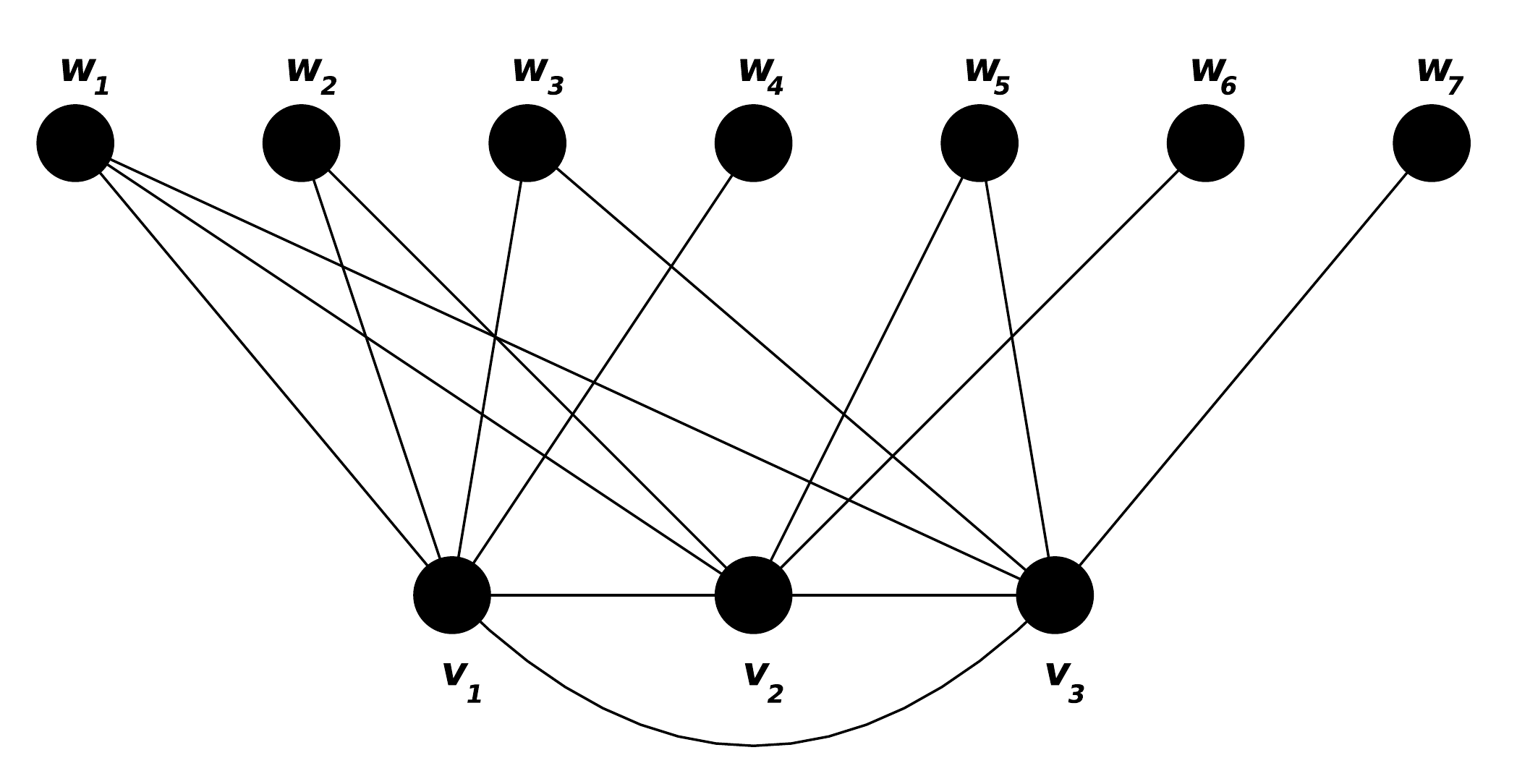}
	\caption{Construction of graph $G$ for $k=3$.}
	\label{contraejGraph}
\end{figure}

Now, $d(v_i)=2^{k-1}+k-1$ for all $i=1,\ldots,k$ and $N(v_i) \neq N(v_j)$ for all $1 \leq i \neq j \leq k$.
Then $\{v_i\} \cup N(v_i)$ is a $v_i$-star biclique.
Also, as $N(v_i) - N(v_j) \neq \emptyset$ for all $1 \leq i \neq j \leq k$, we have that 
$\{v_i\} \cup \{v_j\} \cup (N(v_i)-N(v_j))$ is a $v_i$-star biclique.
We can conclude that $G$ has $k + k(k-1) = k^2$ bicliques.
Finally if $n=k + 2^k -1$, we can see that for $k=6$ (i.e. $n=69)$, we have $35=\lceil \frac{n}{2} \rceil < k^2=36$ 
but for $k=7$ (i.e. $n=134)$, we have $67=\lceil \frac{n}{2} \rceil > k^2=49$.
In fact, for $k \in  \mathbb{R}$, $k \geq 6.13$, (i.e. $n \geq 75)$, we have that $\lceil \frac{n}{2} \rceil \geq k^2$.
Following the idea of this example,  we state the following conjecture.

\begin{conjecture}
Let $G$ be a false-twin-free graph on $n\geq 2$ vertices.\\ If $k + 2^k -1 \leq n < (k+1) + 2^{k+1} -1$, for $k \in \mathbb{N}$,
then:\\ If $n \leq 75$, then $G$ has at least $\lceil \frac{n}{2} \rceil$ bicliques, otherwise it has at least $k^2$ bicliques. 
\end{conjecture}

Moreover, as it is possible to extend this idea to construct a graph $G$ on $n$ vertices, for each $n$, 
$k + 2^k -1 \leq n \leq (k+1) + 2^{k+1} -1$, and containing from $k^2$ to $(k+1)^2$ bicliques, we present this tighter conjecture.

\begin{conjecture}
Let $G$ be a false-twin-free graph on $n\geq 2$ vertices.\\ If $k + 2^k -1 \leq n < (k+1) + 2^{k+1} -1$, for $k \in \mathbb{N}$,
then:\\ If $n \leq 75$, then $G$ has at least $\lceil \frac{n}{2} \rceil$ bicliques, otherwise it has at least 
$k^2+ \lfloor(2k+1)\frac{n-(k + 2^k -1)}{2^k+1}\rfloor$ bicliques. 
\end{conjecture}

To finish this section, we present the following results about the structure of false-twin-free graphs.

\begin{lemma}
Let $G$ be a false-twin-free graph. If $G$ has a $K_3$ as a subgraph then there is no vertex that belongs to all bicliques.
\end{lemma}

\begin{lemma}
Let $G$ be a false-twin-free graph. There are at most two vertices $v,w$ that belong to all bicliques and they must be adjacent. Moreover,
for every other vertex $u$, $u$ is adjacent to $v$ if and only if $u$ is not adjacent to $w$.
\end{lemma}

\begin{lemma}
Let $G$ be a false-twin-free graph. For every biclique $B$ there exists at most one vertex that belongs only to $B$.
\end{lemma}

\begin{lemma}
Let $G$ be a false-twin-free graph. Let $v$ a vertex such that $d(v) \geq 2$. Then $v$ belongs to at least two different bicliques.
\end{lemma}

From last lemma, we obtain this immediate result.

\begin{corollary}
Let $G$ be a false-twin-free graph. There are at least $\lceil \frac{n}{2} \rceil$ vertices that belong to at least two bicliques.
\end{corollary}

\begin{lemma}
Let $G$ be a false-twin-free graph, $G \neq K_2$. If there are two vertices of degree one, they belong to different bicliques.
\end{lemma}

\section{Conclusions}

In this paper we study bounds on the minimum number of bicliques in a graph.  Since adding false-twin vertices to a graph does not change 
the number of bicliques, we restrict to false-twins-free graphs. We give a tight lower bound for a subclass of $\{C_4$,false-twin$\}$-free 
graphs and for the class of $\{K_3$,false-twin$\}$-free graphs. Also we discuss the problem for general false-twin-graphs showing that 
this bound does not hold. Finally, we present two conjectures for bounds in general false-twin-free graphs.

\bibliography{biblio}

\end{document}